 \documentclass[accepted]{uai2023} 

\usepackage[american]{babel}
\usepackage{natbib} 
\usepackage{mathtools} 
\usepackage{booktabs} 
\usepackage{tikz} 
\usepackage{caption}
\usepackage{subcaption} 
\usepackage{amsfonts} 
\usepackage{amsthm}
\usetikzlibrary{arrows.meta} 
\usepackage{algorithm}
\usepackage{algpseudocode}
\usepackage{array}
\usepackage{multirow}

\DeclareMathOperator{\circlearrow}{\hbox{$\circ$}\kern-1.5pt\hbox{$\rightarrow$}}
\DeclareMathOperator{\circlecircle}{\hbox{$\circ$}\kern-1.2pt\hbox{$--$}\kern-1.5pt\hbox{$\circ$}}
\DeclareMathOperator{\diedgeright}{\textcolor{blue}{\boldsymbol{\rightarrow}}}
\DeclareMathOperator{\diedgeleft}{\textcolor{blue}{\boldsymbol{\leftarrow}}}
\DeclareMathOperator{\biedge}{\textcolor{red}{\boldsymbol{\leftrightarrow}}}

\DeclareMathOperator{\pa}{pa}

\DeclareMathOperator{\odds}{\text{OR}}

\def\ci{\perp\!\!\!\perp}

\newcommand{\E}{\mathbb{E}}
\newcommand{\G}{\mathcal{G}}

\newtheorem{theorem}{Theorem}

\title{Causal Inference With Outcome-Dependent Missingness And Self-Censoring}

\author[1]{\href{mailto:<jmc8@williams.edu>?Subject=Your UAI 2023 paper}{Jacob M. Chen}}
\author[2]{\href{mailto:<d.malinsky@columbia.edu>?Subject=Your UAI 2023 paper}{Daniel Malinsky}}
\author[1]{\href{mailto:<rb17@williams.edu>?Subject=Your UAI 2023 paper}{Rohit Bhattacharya}}
\affil[$\text{\textcolor{white}{1}}$]{%
\texttt{\hspace{-0.6cm} jmc8@williams.edu \hspace{0.4cm} d.malinsky@columbia.edu \hspace{0.6cm} rb17@williams.edu \vspace{0.25cm}}
}
\affil[1]{%
    Department of Computer Science\\
    Williams College
}
\affil[2]{%
    Department of Biostatistics\\
    Columbia University
}

\begin{document}
\maketitle

\begin{abstract}
We consider missingness in the context of causal inference when the outcome of interest may be missing. If the outcome directly affects its own missingness status, i.e., it is ``self-censoring'', this may lead to severely biased causal effect estimates. \cite{miao2015identification} proposed the shadow variable method to correct for bias due to self-censoring; however, verifying the required model assumptions can be difficult. Here, we propose a test based on a randomized incentive variable offered to encourage reporting of the outcome that can be used to verify identification assumptions that are sufficient to correct for both self-censoring and confounding bias. Concretely, the test confirms whether a given set of pre-treatment covariates is sufficient to block all backdoor paths between the treatment and outcome as well as all paths between the treatment and missingness indicator after conditioning on the outcome. We show that under these conditions, the causal effect is identified by using the treatment as a shadow variable, and it leads to an intuitive inverse probability weighting estimator that uses a product of the treatment and response weights. We evaluate the efficacy of our test and downstream estimator via simulations.
\end{abstract}

\section{Introduction}
\label{sec:intro}

``Self-censoring'' is a type of missingness-not-at-random (MNAR) phenomenon that poses a particularly difficult obstacle to valid inference in settings with missing data. Researchers have begun to address MNAR missing data problems, i.e., problems where the probability of missingness depends on variables that exhibit missingness themselves. However, the task of mitigating bias due to self-censoring remains relatively unexplored \citep{mohan2021graphical}. Here, we consider outcome-dependent self-censoring---situations where the outcome determines its own missingness---in the context of computing causal effects. This kind of self-censoring is quite common in practice, for example, in studies where the outcome is an attribute associated with social stigma (such as drug use or risky sexual behaviors) or in general in settings where the outcome is ascertained by voluntary survey response.

Recent work in missing data uses directed acyclic graphs (DAGs) to represent substantive assumptions about causal relations among variables, including indicators of missingness \citep{daniel2012using, mohan2013missing, mohan2021graphical}. Within this framework, \cite{bhattacharya2019mid} and \cite{nabi2020mid} derive a sound and complete criterion for when it is possible to recover the full data law of a missing data DAG model in the presence of MNAR missingness and unmeasured confounding. However, this criterion specifically excludes self-censoring, i.e., self-censoring prevents non-parametric identification of the full data law unless additional (non-structural) assumptions are made about the data generating process.

When the full structure of the DAG is unknown, the possibility of unmeasured confounding, or the existence of latent variables, poses a further challenge to observational causal inference as this complicates the process of finding a valid set of covariates to adjust for \citep{shpitser2012validity}. Methods for covariate adjustment under MNAR data have been proposed by \cite{saadati2019adjustment} and \cite{yang2019causal}; however, these methods require full knowledge of the structure of the missing data DAG and do not allow for self-censoring on the outcome. \cite{nabi2022testability} propose empirical tests to verify assumptions encoded in certain subclasses of MNAR models, but all of these also exclude the possibility of self-censoring. To recover the full data law given a self-censoring outcome and under a certain completeness condition, \cite{miao2015identification} propose the use of a \emph{shadow variable}, a variable that satisfies a relevance assumption with respect to the outcome and an exclusion restriction with respect to the outcome's missingness indicator. However, empirical tests to verify the shadow variable assumptions as well as the validity of downstream covariate adjustment remains an open problem that we seek to address in this work.

Here, we build on the work by \cite{miao2015identification} on identification under self-censoring and \cite{entner2013data} on covariate selection to propose a method for simultaneously recovering the target law -- the joint distribution of variables as if there were no missingness -- and a valid covariate adjustment set in the presence of a self-censoring outcome. 

We propose a test, based on a randomized incentive variable offered to encourage reporting of the outcome, which can be used to verify identification assumptions that are sufficient to correct for both self-censoring and confounding bias. The test proceeds in two stages. The first stage confirms dependence between the incentive for reporting and the missingness indicator of the outcome. The second stage executes a search for a pre-treatment covariate and a covariate adjustment set that satisfies three conditions: (i) independence between the treatment and the incentive conditional on the outcome, missingness indicator of the outcome, and the covariate adjustment set, (ii) dependence between the pre-treatment covariate and the missingness indicator of the outcome while conditioning on the covariate adjustment set, and (iii) independence between the same variables and adjustment set in test (ii) while additionally conditioning on the treatment. We prove that, when these conditions hold, one may use the treatment as a valid shadow variable and that there is a valid covariate adjustment set. We derive the corresponding identifying functional for the causal effect and propose an intuitive inverse probability weighting estimator that uses a product of the treatment and response weights. We evaluate the efficacy of our method via simulations.

\textbf{Related work on self-censoring}: \cite{sportisse2020estimation} propose an imputation method for self-censored data that assumes factorization according to a certain latent variable DAG and parametric models for the missingness process. \cite{mohan2018estimation} propose methods for recovery of the full data law when all variables are discrete and when certain matrices corresponding to conditional probability tables are invertible. \cite{duarte2021automated} propose an algorithm for computing bounds on the causal effect in the discrete setting, which may converge to point identification in certain cases. \cite{d2010new} and \cite{tchetgen2017general} propose instrumental variable methods that place homogeneity restrictions on the missingness process in addition to requiring the presence of a valid instrument. Though the randomized incentive variable we consider could be used as an instrumental variable, the homogeneity restriction is untestable and can be restrictive in many real-world settings. To our knowledge, the testability of identifying assumptions of the shadow variable method and covariate adjustment under self-censoring has not been explored before. 

\section{Motivating Example}
\label{sec:motivating_example}

For our motivating example, we describe a hypothetical study inspired by \cite{turner2009improving} for evaluating the effect of public health programs that encourage safe sex practices. Consider an observational study where researchers offer a sexual education program to encourage condom use. That is, enrollment into this program is not randomized. In a follow-up survey, participants may also choose not to disclose their post-program condom use habits, precisely because of the opinions they hold on the practice. Hence, to estimate the causal effect of the program on improving condom use, the researchers must overcomes challenges related to both confounding and self-censoring.

To incentivize response to answer sensitive questions regarding sexual behavior, the researchers randomly assigned study participants to be surveyed via a phone interview conducted by a real human being or an automated program called the Telephone Audio Computer-Assisted Self-Interviewing (T-ACASI) program. \cite{turner2009improving} applied this randomized incentive strategy to successfully increase response rates on questions pertaining to drug use or risky sexual behaviors in a study conducted in Baltimore, USA. 
One can also imagine other randomized incentives offered to increase response\footnote{The effect of the incentive on response need not be monotonic. Our method simply relies on some correlation between the incentive and response, which we can test using the observed data.}, such as random lotteries for gift cards. Such incentives are often offered as part of studies that have a survey component. However, in most cases, the incentives can only encourage response but not guarantee it, and so the issue of self-censoring still persists.

In this work, we aim to simultaneously address the challenges of covariate selection and overcoming self-censoring, or non-response bias, inherent in observational studies that, for example, ask respondents sensitive questions.

\section{Model and Problem Setup}
\label{sec:preliminaries}

We assume the causal structure of the system is represented via a directed acyclic graph $\G$ defined over a set of vertices ${\bf V} = \{A, Y^{(1)}, R_Y, Y, I\} \cup {\bf W} \cup {\bf U}$, where $A$ represents the treatment variable, $Y^{(1)}$ represents the outcome had we -- possibly contrary to fact -- been able to observe it, $R_Y$ represents the corresponding binary missingness indicator for the outcome, $Y$ represents the factual observed outcome -- which may be either a numeric value or ``$?$'' if the observation is missing, $I$ represents the incentive variable -- a randomized variable that affects whether or not an individual responds and we observe their corresponding outcome, ${\bf W}$ represents an observed set of pre-treatment covariates, and ${\bf U}$ denotes a set of unobserved (latent) covariates. The entire set ${\bf V}$ is assumed to be causally sufficient, i.e., there are no additional unmeasured common causes of any two variables in ${\bf V}$. 

The above setup describes a causal Bayesian network over the variables ${\bf V}$ \citep{spirtes2000causation, pearl2009causality}, with an added restriction on the relationship between the missing variable $Y^{(1)}$, its missingness indicator, and the observed outcome: $Y = Y^{(1)}$ when $R_Y=1$ and $Y=$ $?$ when $R_Y=0$. This is the missing data version of the \emph{consistency} assumption in causal inference \citep{nabi2022causal}. Due to the fixed deterministic nature of this relationship, we omit drawing the observed outcome $Y$ in some of our figures.

Independences implied in the full data law $p({\bf V})$ can be read off from $\G$ via the d-separation criterion \citep{pearl1988probabilistic, mohan2021graphical}.
Here, we make a \emph{faithfulness} assumption \citep{spirtes2000causation}, which states that any independencies in the distribution $p({\bf V})$ must correspond to d-separation statements in $\G$. Formally, $A \ci B \mid {\bf C}$ in $p({\bf V})$ if and only if  $A \ci_{\text{d-sep}} B \mid {\bf C}$ in $\G$. For tests that involve conditioning on the missing outcome $Y^{(1)}$, we also assume an extension of faithfulness used in causal discovery procedures for missing data settings \citep{tu2019causal}. This assumption states that any independences that exist conditional on $R_Y$ in $p(\textbf{V})$ must also hold in the observed data, i.e., conditional on $R_Y=1$. 
Formally, we assume $A \ci B \mid {\bf C}, R_Y$ in $p(\textbf{V})$ if and only if $A \ci B \mid {\bf C}, R_Y=1$ in $p(\textbf{V})$ when $Y^{(1)}$ is one of $A$, $B$, or an element of the conditioning set ${\bf C}$.

The above assumptions are commonly used across most graphical model selection procedures \citep{spirtes2000causation}. We now list and provide brief justification for additional structural assumptions important for our method. Let $\pa_\G(V)$ denote the parents of the variable $V$ in the graph $\G$. We assume that the data are generated from a distribution $p({\bf V})$ that is Markov and faithful with respect to a DAG $\G$ that satisfies the structural assumptions M1-M4.
\begin{enumerate}
    \item[(M1)] The only parents of $Y$ are $Y^{(1)}$ and $R_Y$, i.e., $\pa_\G(Y) = \{Y^{(1)}, R_Y\}$.
    \item[(M2)] There is an edge from $A$ to $Y^{(1)}$ and an edge from $Y^{(1)}$ to $R_Y$ (i.e. the causal path $A\diedgeright Y^{(1)} \diedgeright R_Y$ exists in the graph).
    \item[(M3)] The incentive $I$ is randomly assigned (i.e. $\pa_\G(I)=\emptyset$) and may only be a parent of the missingness indicator $R_Y$.
    \item[(M4)] $Y$ is not a parent of any variables in ${\bf V}$ and does not have any children, $R_Y$ is not a parent of any variables in ${\bf V} \setminus \{Y\}$, $A$ is not a parent of any variables in ${\bf W} \cup {\bf U}$, and $Y^{(1)}$ cannot be a parent of $A$ nor any variables in ${\bf W} \cup {\bf U}$. That is, we have an ordering, $\{I\} \cup {\bf W} \cup {\bf U} < A < Y^{(1)} < R_Y < Y$.
\end{enumerate}
Assumption M1 and disallowing $R_Y$ from having any children aside from $Y$ in assumption M4 are standard restrictions in missing data DAG models \citep{mohan2013missing}. We require assumption M2 as it simplifies our empirical tests. However, M2 is a relatively mild assumption as the existence of these edges is the primary motivation for applying our method. Assumption M3 makes sure that $I$ is a valid proxy variable for designing indirect tests about the validity of the treatment as a shadow variable, which is described in Section~\ref{sec:identification}. Finally, assumption M4 states that ${\bf W} \cup {\bf U}$ are all pre-treatment variables.

Figure \ref{fig:assumptions} graphically displays assumptions M1-M4. The red edges are assumed to exist while the blue edges may or may not exist. We draw $Y$ here to illustrate M1, but we will omit it and the red dashed edges in all figures going forward (due to the deterministic nature of its relation with $Y^{(1)}$ and $R_Y$.)
\begin{figure}[ht]
    \centering
    \scalebox{0.8}{
	    \begin{tikzpicture}[>=stealth, node distance=1.5cm]
			\tikzstyle{square} = [draw, thick, minimum size=1.0mm, inner sep=3pt]
			\begin{scope}
				\path[->, very thick]
				node[] (a) {$A$}
                node[right of=a] (y) {$Y^{(1)}$}
                node[right of=y] (ry) {$R_Y$}
                node[below of=y, xshift=0.75cm] (ystar) {$Y$}
                node[right of=ry] (i) {$I$}
                node[above of=y] (wu) {${\bf W} \cup {\bf U}$}
                
                (a) edge[red] (y)
                (y) edge[red] (ry)
                (y) edge[red, dashed] (ystar)
                (ry) edge[red, dashed] (ystar)
                (i) edge[blue] (ry)
                (wu) edge[blue] (a)
                (wu) edge[blue] (y)
                (wu) edge[blue] (ry)
				;
			\end{scope}
		\end{tikzpicture}
    }
    \caption{Graph depicting assumptions M1-M4.}
    \label{fig:assumptions}
\end{figure}
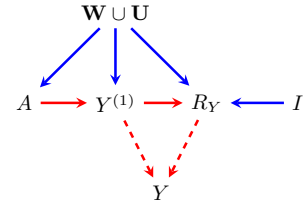

We briefly connect our notation and problem setup back to the initial motivating example. The variable $A$ represents whether or not an individual enrolls in the sexual education program. The variable $Y^{(1)}$ represents an individual's (true) post-program condom use habits. 
The variable $R_Y$ represents whether an individual reports their condom use habits. The variable $I$ represents whether a participant was interviewed by the T-ACASI program or by a human interviewer. Finally, the variables ${\bf W} \cup {\bf U}$ represent fully observed and unobserved covariates in the problem, respectively. We now formally define our target causal parameter. 

\subsubsection*{Target of Inference}
Under missingness, a causal effect of the treatment on the outcome corresponds to a contrast between  potential outcomes $Y^{(A=a, R_Y=1)}$ and $Y^{(A=a', R_Y=1)}$, where $Y^{(A=a, R_Y=1)}$ denotes the value of the outcome had the treatment been set to some value $a$ via intervention and had the outcome been observed. Moving forward, we use $Y^{(a, 1)}$ and $Y^{(a', 1)}$ for brevity. Our target of inference is the average causal effect (ACE), i.e., the mean difference $\E[Y^{(a, 1)} - Y^{(a', 1)}]$.\footnote{Similar notation for representing potential outcomes under missingness has been used in \cite{nabi2022causal}. This can also be expressed using do-notation as in \cite{saadati2019adjustment}.} In the next section we derive an identification formula for the ACE in terms of the observed data distribution $p(A, Y, R_Y, {\bf Z})$ where ${\bf Z} \subset {\bf W}$ based on the backdoor adjustment formula from \cite{pearl1995causal} and the shadow variable method proposed by \cite{miao2015identification}.

\section{Identification}
\label{sec:identification}

In this section we demonstrate recovery of the ACE assuming we are given a valid backdoor adjustment set and shadow variable adjustment set ${\bf Z} \subset {\bf W}$. 
Note that we use a strict subset relation as we will use at least one of the remaining pre-treatment covariates for verification of the identifying assumptions. We can also ignore the incentive variable $I$ here because it does not play a role in identification, only in the verification of the assumptions surrounding ${\bf Z}$ later on.

In the following, we focus on identification of $\E[Y^{(a,1)}]$; identification of $\E[Y^{(a',1)}]$ follows similarly. We perform identification in two steps. In the first step, we assume we have access to the underlying counterfactual $Y^{(1)}$. A set ${\bf Z}$ is said to be a valid backdoor adjustment set relative to the treatment $A$ and outcome $Y^{(1)}$ if
\begin{enumerate}
    \item[(B1)] ${\bf Z}$ does not contain any variables that are causal descendants of $A$, and
    \item[(B2)] $A$ and $Y^{(1)}$ are d-separated when conditioning on ${\bf Z}$ in a modified graph where all outgoing edges from $A$ are deleted.
\end{enumerate}
If ${\bf Z}$ is a valid backdoor adjustment set, then we have the following backdoor adjustment formula in terms of the full data law involving $Y^{(1)}$ \citep{pearl1995causal}.
\begin{align}
    \E[Y^{(a, 1)}] &= \sum_{{\bf Z}} \E[Y^{(1)} \mid A=a, {\bf Z}] \times p({\bf Z})
    \label{eq:backdoor}
\end{align}
However, due to missingness, we have access to only the margin $p(A, Y, R_Y, {\bf Z})$, which does not include $Y^{(1)}$. In our problem setting, the counterfactual expression $\E[Y^{(1)} \mid A=a, {\bf Z}]$ is not equal to the observed regression $ \E[Y \mid A=a, {\bf Z}, R_Y=1]$ since $Y^{(1)} \not\ci R_Y \mid A, {\bf Z}$ from assumption M2, i.e., due to self-censoring.

Next, we define shadow variables as described by \cite{miao2015identification}, and use it to overcome the issue of self-censoring. A variable $S$ is a valid \textit{shadow variable} if it is a fully observed variable that satisfies the following independence relations.
\begin{enumerate}
    \item[(S1)] $S \not \ci Y^{(1)} \mid R_Y=1, {\bf Z}$, and
    \item[(S2)] $S \ci R_Y \mid Y^{(1)}, {\bf Z}$.
\end{enumerate}
A shadow variable $S$ helps us recover the target law as if there were no missingness by enabling identification of the propensity score of the missingness indicator $p(R_Y=1 \mid A, Y^{(1)}, {\bf Z})$ from the margin $p(A, Y, R_Y, {\bf Z})$. Intuitively, shadow variables have some non-zero effect on the self-censoring variable, which means they contain some useful information about it. In addition, they are independent of the missingness mechanism, which makes it possible for us to use the shadow variable to infer information about the self-censoring variable using just observed rows of data. We refer to the set ${\bf Z} \subset {\bf W}$ that satisfies conditions S1 and S2 as the {\it shadow variable adjustment set}. In our method, we verify whether the treatment $A$ is a valid shadow variable and directly use it as such. 

To identify \eqref{eq:backdoor}, it is sufficient to identify the joint distribution of $p(A, Y^{(1)}, {\bf Z})$. By the chain rule of probability we have,
\begin{align}
    p(A, Y^{(1)}, {\bf Z}) = \frac{p(A, Y^{(1)}, {\bf Z}, R_Y=1)}{p(R_Y = 1 \mid A, Y^{(1)}, {\bf Z})}.
    \label{eq:chain_rule}
\end{align}

The numerator is a function of observed data due to consistency. Identification of this joint then reduces to identification of the propensity score $p(R_Y = 1 \mid A, Y^{(1)}, {\bf Z})$. If we are able to verify that $A$ is a valid shadow variable, then the independence $A \ci R_Y \mid Y^{(1)}, {\bf Z}$ will hold from S2. Hence, $p(R_Y=1 \mid A, Y^{(1)}, {\bf Z}) = p(R_Y=1 \mid Y^{(1)}, {\bf Z})$.\footnote{Note that from \eqref{eq:chain_rule} we identify the full target law. However, there may be cases when recovery of the full target law is not necessary to recover the causal effect of interest.}

Following \cite{miao2015identification}, we use an odds ratio factorization of the propensity score to perform the identification of $p(R_Y=1 \mid Y^{(1)}, {\bf Z})$.
For two variables $X, Y$ and a set of variables ${\bf Z}$, the conditional odds ratio function is defined as
\begin{align*}
    \odds(X, Y &\mid {\bf Z}) = \\
    &\frac{p(X \mid Y, {\bf Z})}{p(X=x_0 \mid Y, {\bf Z})} \times \frac{p(X=x_0 \mid Y=y_0, {\bf Z})}{p(X \mid Y=y_0, {\bf Z})}, \nonumber
\end{align*}
where $x_0$ and $y_0$ are the specified reference values for $X$ and $Y$. From the above definition, the odds ratio is 1 whenever $X=x_0$ or $Y=y_0$. Furthermore, the odds ratio is 1 for all values of $X, Y, Z$ if and only if $X \ci Y \mid {\bf Z}$. 

Without loss of generality, we pick any arbitrary value $y_0$ in the state space of $Y^{(1)}$ to be the reference value for the outcome and $R_Y=1$ to be the reference value of the missingness indicator. Let $\pi_0({\bf Z})$ denote the propensity score for $R_Y$ at the reference value $y_0$, i.e., $\pi_0({\bf Z}) \coloneqq p(R_Y =1 \mid Y^{(1)} = y_0, {\bf Z})$. Let $\eta(Y^{(1)}, {\bf Z})$ denote the conditional odds ratio function relating $Y^{(1)}$ and $R_Y$ at values where $R_Y=0$, i.e., $\eta(Y^{(1)}, {\bf Z}) \coloneqq \odds(R_Y=0, Y^{(1)} \mid {\bf Z})$. Then the odds ratio factorization of the propensity score can be written as (to keep our results self-contained, a proof of this factorization is provided in the Appendix),
\begin{align}
    p(R_Y = 1 &\mid Y^{(1)}, {\bf Z}) = \nonumber \\
    &\frac{\pi_0({\bf Z})}{\pi_0({\bf Z}) + \eta(Y^{(1)}, {\bf Z})(1-\pi_0({\bf Z}))}
    \label{eq:or_factorization}
\end{align}
Whenever $Y^{(1)} = y_0$, $\eta(Y^{(1)}, {\bf Z})=1$, and we have that $p(R_Y=1 \mid Y^{(1)}=y_0, {\bf Z}) = \pi_0({\bf Z})$ because the denominator simplifies to a value of $1$. At any other value of $Y^{(1)}$, $\eta(Y^{(1)}, {\bf Z}) \neq 1$, and the odds ratio factorization of the propensity score will return a different value.

When $A$ is a valid shadow variable and the conditional distribution $p(Y^{(1)} \mid R_Y=1, A, {\bf Z})$ satisfies a widely used completeness condition\footnote{For all square-integrable functions $h(A, Y^{(1)})$, $\E[h(A, Y^{(1)}) \mid R_Y=1, A, {\bf Z}] = 0$ almost surely if and only if $h(A, Y^{(1)}) = 0$ almost surely (see \cite{newey2003instrumental} or Section 3 of \cite{miao2015identification} for more details).}, \cite{miao2015identification} show that $\pi_0$ and $\eta$ (and hence the  propensity score) are identified from the observed data.

Consider Figure \ref{fig:adjustment_sets}, which demonstrates how certain sets of pre-treatment covariates might be sufficient to adjust for confounding but not missingness or vice versa. In this DAG, the set $\{W_1\}$ satisfies B1 and B2. However, it does not satisfy S1 and S2 when considering the treatment $A$ as the shadow variable because of the open collider at $Y^{(1)}$ on the path $A \diedgeright Y^{(1)} \diedgeleft W_2 \diedgeright R_Y$. On the other hand, the set $\{W_2\}$ satisfies S1 and S2 when considering the treatment $A$ as the shadow variable, but the backdoor path through $W_1$ remains open. Only sets that include both $W_1$ and $W_2$ such as ${\bf Z} = \{W_1, W_2\}$ satisfy all of B1, B2, S1, and S2 simultaneously. We now present our main identification results that relies on a set of covariates that satisfy all these conditions simultaneously.
\begin{figure}[ht]
    \centering
    \scalebox{0.8}{
	    \begin{tikzpicture}[>=stealth, node distance=1.5cm]
			\tikzstyle{square} = [draw, thick, minimum size=1.0mm, inner sep=3pt]
			\begin{scope}
				\path[->, very thick]
				node[] (a) {$A$}
                node[right of=a] (y) {$Y^{(1)}$}
                node[right of=y] (ry) {$R_Y$}
                node[right of=ry] (i) {$I$}
                node[above of=a, xshift=0.75cm] (w1) {$W_1$}
                node[above of=y, xshift=0.75cm] (w2) {$W_2$}
                node[above of=a, xshift=-0.75cm] (w3) {$W_3$}
                
                (a) edge[blue] (y)
                (y) edge[blue] (ry)
                (i) edge[blue] (ry)
                (w1) edge[blue] (a)
                (w1) edge[blue] (y)
                (w2) edge[blue] (y)
                (w2) edge[blue] (ry)
                (w3) edge[blue] (a)
				;
			\end{scope}
		\end{tikzpicture}
    }
    \caption{Graph demonstrating how different subsets of {\bf W} satisfy the identifying assumptions B1, B2, S1, and S2.
    }
    \label{fig:adjustment_sets}
\end{figure}
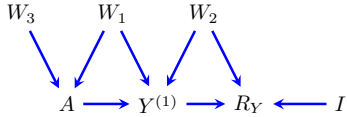
\begin{theorem}
    Under structural assumptions M1-M4 and completeness of $p(Y^{(1)} \mid R_Y=1, A, {\bf Z})$, if ${\bf Z}$ satisfies B1 and B2 and $A$ is a valid shadow variable satisfying S1 and S2 conditional on ${\bf Z}$, then the expected value of the counterfactual outcome $\E[Y^{(a,1)}]$ is identified from the observed data distribution $p(A, Y, {\bf Z}, R_Y)$ as follows:
    \begin{align}
        \E[Y^{(a,1)}] = \E\bigg[\frac{R_Y \times \mathbb{I}(A=a) \times Y}{p(R_Y=1 | Y^{(1)}, {\bf Z}) \times p(A=a | {\bf Z})}\bigg]
        \label{eq:identification}
    \end{align}
    \label{thm:identification}
\end{theorem}
A full proof of Theorem \ref{thm:identification} is provided in the Appendix, but the intuition is simple. The use of $R_Y$ in the numerator of \eqref{eq:identification} ensures that we use only observed rows of data. Further, under our assumptions, both propensity scores in the denominator are identified as functions of the observed data; $p(A=a \mid {\bf Z})$ depends on only observed data quantities and $p(R_Y=1 \mid Y^{(1)}, {\bf Z})$ is identified via the odds ratio factorization in \eqref{eq:or_factorization} with the treatment as a shadow variable. We can then apply standard laws of probability to show that \eqref{eq:identification} is equivalent to the full data law adjustment functional in \eqref{eq:backdoor}. Since ${\bf Z}$ is assumed to be a valid backdoor adjustment set, this further implies equivalence of \eqref{eq:identification} to the counterfactual mean $\E[Y^{(a, 1)}]$. 

The above identification argument relies on the absence of certain edges from the graph $\G$ in order to satisfy assumptions S1, S2, B1, and B2. Figure \ref{fig:invalidating_edges} shows the edges that, if present, preclude identification. In the figure we use bidirected edges as a shorthand for a path involving unmeasured variables. For example, $A \biedge Y^{(1)}$ is shorthand for the presence of a d-connecting path $A \diedgeleft U_1 \cdots U_k \diedgeright Y^{(1)}$, where $U_1,\dots, U_k$ are unmeasured variables in ${\bf U}$. In the figure, it is possible for the treatment $A$ and a set ${\bf Z} \subset {\bf W}$ to satisfy the shadow variable conditions S1 and S2 whenever the red dashed edges are absent. Similarly, it is possible for ${\bf Z} \subset {\bf W}$ to satisfy the backdoor conditions B1 and B2 whenever the green dashed edge is absent. In essence, our method in the next section tests for the absence of these edges.
\begin{figure}[ht]
    \centering
    \scalebox{0.8}{
	    \begin{tikzpicture}[>=stealth, node distance=1.75cm]
			\tikzstyle{square} = [draw, thick, minimum size=1.0mm, inner sep=3pt]
			\begin{scope}
				\path[->, very thick]
				node[] (a) {$A$}
                node[right of=a] (y) {$Y^{(1)}$}
                node[right of=y] (ry) {$R_Y$}
                node[right of=ry] (i) {$I$}
                node[above of=y] (w) {${\bf W}$}
                
                (a) edge[blue] (y)
                (y) edge[blue] (ry)
                (i) edge[blue] (ry)
                (w) edge[blue] (a)
                (w) edge[blue] (y)
                (w) edge[blue] (ry)
                (a) edge[red, dashed, bend right] (ry)
                (a) edge[red, dashed, bend right, <->, in=240, out=300] (ry)
                (a) edge[green, dashed, bend left, <->] (y)
                (y) edge[red, dashed, bend left, <->] (ry)
				;
			\end{scope}
		\end{tikzpicture}
    }
    \caption{Red and green dashed edges impede identification.}
    \label{fig:invalidating_edges}
\end{figure}
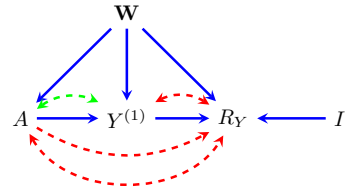

\section{Tests For Identification Conditions}
\label{sec:verification}

\cite{entner2013data} propose a two-stage test to verify whether a set of pre-treatment covariates ${\bf Z}$ satisfy B1 and B2 for backdoor adjustment when there is no missing data or sample selection bias. However, from the previous section we have seen that, to overcome self-censoring, ${\bf Z}$ must also satisfy the shadow variable conditions S1 and S2 for some candidate shadow variable. Testing S2 using observed data is impossible in general, though, as such a test involves conditioning on the counterfactual variable $Y^{(1)}$. Here, we augment the test from \cite{entner2013data} so that it is capable of testing for the backdoor adjustment conditions in addition to the shadow variable assumptions when considering the treatment $A$ as a shadow variable and only using observed data quantities. This allows us to validate the crucial identifying assumptions laid out in the previous section.

\subsection{Method}
\label{subsec:method}

The two-stage method described in \cite{entner2013data} searches for a $W \in {\bf W}$ and ${\bf Z} \subseteq {\bf W} \setminus \{W\}$ such that (i) $W \not \ci Y \mid {\bf Z}$ and (ii) $W \ci Y \mid {\bf Z}, A$. They assume the same partial order of variables as in assumption M4 except that the variables $I$ and $R_Y$ are not present in the graph. In addition, $Y$ is a fully observed variable. \cite{entner2013data} proved that, when conditions (i) and (ii) hold, {\bf Z} is a valid backdoor adjustment set for the causal effect of $A$ on $Y$.

For our method, we propose to test if
\begin{enumerate}
    \item[(C1)] $I \not \ci R_Y$ 
\end{enumerate}
and then search for some $W \in {\bf W}$ and
${\bf Z} \subseteq {\bf W} \setminus \{W\}$ 
such that
\begin{enumerate}
    \item[(C2)] $A \ci I \mid Y^{(1)}, R_Y=1, {\bf Z}$
    \item[(C3)] $W \not \ci R_Y \mid {\bf Z}$
    \item[(C4)] $W \ci R_Y \mid A, {\bf Z}$ 
\end{enumerate}

Conditions C1 and C2 are used to confirm the absence of the red dashed edges in Figure~\ref{fig:invalidating_edges} that would prevent us from using $A$ as a valid shadow variable; conditions C3 and C4 are a modification of the tests from \cite{entner2013data} that indirectly confirm the validity of ${\bf Z}$ as a backdoor adjustment set for the counterfactual outcome by using the missingness indicator to formulate the tests instead. This is formalized in the theorem below.
\begin{theorem}
    When C1-C4 hold in a distribution $p({\bf V})$ that is Markov and faithful w.r.t a causal DAG $\G$ satisfying assumptions M1-M4, ${\bf Z}$ is a valid backdoor adjustment set for the causal effect of $A$ on $Y^{(1)}$, and $A$ is a valid shadow variable with ${\bf Z}$ as a shadow variable adjustment set.
\end{theorem}
\begin{proof}
    From \cite{entner2013data}, if C3 and C4 hold, {\bf Z} is a valid backdoor adjustment set for computing the effect of $A$ on $R_Y$. From assumptions M2 and M4, all backdoor paths between $A$ and $Y^{(1)}$, i.e., $A\diedgeleft \cdots \diedgeright Y^{(1)}$, also extend into a backdoor path $A\diedgeleft \cdots \diedgeright Y^{(1)} \diedgeright R_Y$ to $R_Y$. Thus, if ${\bf Z}$ blocks all backdoor paths between $A$ and $R_Y$, it also blocks all backdoor paths between $A$ and $Y^{(1)}$, making ${\bf Z}$ a valid adjustment set for the effect of $A$ on $Y^{(1)}$.

    Next, we prove that ${\bf Z}$ can be used for  shadow variable adjustment using $A$ as the  shadow variable. Condition S1 holds trivially as we have $A \diedgeright Y^{(1)}$ according to assumption M2. We now prove S2 also holds: that $A \ci R_Y \mid Y^{(1)}, {\bf Z}$. First, under assumption M2, C1 ensures that the randomized incentive $I$ has a directed edge to $R_Y$ and no other outgoing edges. In order for $A$ to be a valid shadow variable, all paths between $A$ and $R_Y$ conditional on $Y^{(1)}$ and ${\bf Z}$ must be blocked. The following 4 cases cover all possible paths between $A$ and $R_Y$ given our assumptions.
    \begin{enumerate}
        \item Causal paths from $A$ to $R_Y$. One possible causal path $A \diedgeright Y^{(1)} \diedgeright R_Y$  exists by assumption M4, but it is blocked by conditioning on $Y^{(1)}$. The second possibility of $A \diedgeright R_Y$ cannot exist because it would imply an open path between $I$ and $A$ conditional on $Y^{(1)}, R_Y=1,$ and ${\bf Z}$, contradicting C2. 
        \item Paths of the form $A \diedgeright Y^{(1)} \diedgeleft \dots \diedgeright R_Y$. Any open paths of this form contradict C2, as it implies the dependence $A \not\ci I \mid Y^{(1)}, R_Y, {\bf Z}$.
        \item Backdoor paths of the form $A \diedgeleft \dots \diedgeright R_Y$ that do not contain $Y^{(1)}$ as a collider. As previously noted, all such backdoor paths are blocked by ${\bf Z}$ based on conditions C3 and C4.
        \item Backdoor paths containing $Y^{(1)}$ as a collider, i.e., $A \diedgeleft \dots \diedgeright Y^{(1)} \diedgeleft \dots \diedgeright R_Y$. All such paths are still blocked despite conditioning on $Y^{(1)}$ as ${\bf Z}$ blocks all backdoor paths between $A$ and $Y^{(1)}$. 
    \end{enumerate}
    Therefore, ${\bf Z}$ is a valid backdoor adjustment set for the causal effect of $A$ on $Y^{(1)}$, fulfilling B1 and B2, and it is also a valid shadow variable adjustment set for $A$ to be a valid shadow variable, fulfilling S1 and S2 for $A$.
\end{proof}

We use Figure \ref{fig:method_example} to illustrate an example of how our tests proceed. Let us begin by only considering the solid blue edges. First, we confirm that C1 holds. Next, none of the conditions C2-C4 can be satisfied by using ${\bf Z}=\emptyset$ and $W=W_i$ for $i=1, 2, 3$. When considering singleton adjustment sets, we get that, for $W=W_3$ and ${\bf Z}=\{W_1\}$, condition C2 is satisfied. However, this set does not satisfy C3 and C4  because of the open collider at $W_1$ that introduces an open backdoor path between $A$ and $Y^{(1)}$. Finally for adjustment sets of size $2$, we get that, when $W=W_3$ and ${\bf Z} = \{W_1, W_2\}$, all conditions are satisfied, providing the correct conclusion that the effect is identified via \eqref{eq:identification} using ${\bf Z} = \{W_1, W_2\}$. 
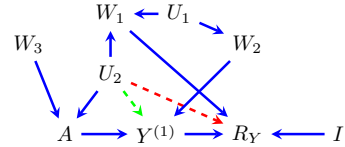
\begin{figure}[ht]
    \centering
    \scalebox{0.8}{
        \begin{tikzpicture}[>=stealth, node distance=1.5cm]
            \tikzstyle{square} = [draw, thick, minimum size=1.0mm, inner sep=3pt]
            \begin{scope}
                \path[->, very thick]
                node[] (a) {$A$}
                node[right of=a] (y) {$Y^{(1)}$}
                node[right of=y] (ry) {$R_Y$}
                node[right of=ry] (i) {$I$}
                node[above of=a, xshift=0.75cm, yshift=0.5cm] (w1) {$W_1$}
                node[right of=w1, xshift=-0.375cm] (u1) {$U_1$}
                node[above of=y, xshift=1.5cm] (w2) {$W_2$}
                node[below of=w1, yshift=0.5cm] (u2) {$U_2$}
                node[above of=a, xshift=-0.6cm] (w3) {$W_3$}
                
                (a) edge[blue] (y)
                (y) edge[blue] (ry)
                (i) edge[blue] (ry)
                (u1) edge[blue] (w1)
                (u1) edge[blue] (w2)
                (u2) edge[blue] (w1)
                (u2) edge[blue] (a)
                (u2) edge[green, dashed] (y)
                (u2) edge[red, dashed] (ry)
                (w1) edge[blue] (ry)
                (w2) edge[blue] (y)
                (w3) edge[blue] (a)
                ;
            \end{scope}
        \end{tikzpicture}
    }
    \caption{$W = W_3$, ${\bf Z} = \{W_1, W_2\}$}
    \label{fig:method_example}
\end{figure}

Next, consider the same DAG in Figure~\ref{fig:method_example} with the green dashed edge being present. Because there is unmeasured confounding between the treatment and outcome variables, there is no possible ${\bf Z}$ that can be a valid backdoor adjustment set. Therefore, C3 and C4 can never be satisfied. Next, let us consider a DAG where the red dashed edge is present. The unmeasured confounding between $A$ and $R_Y$ violates S2 when using $A$ as a shadow variable and also creates a collider path between $A$ and $I$ that ensures that C2 will never be satisfied. In both cases, our method correctly concludes that no adjustment set is possible.

\subsection{Limitations}

There exist DAGs where our identification strategy works but where our method is not able to detect the existence of a valid backdoor and shadow variable adjustment set. Consider Figure \ref{fig:method_fail} where $W_3$ is now a confounder between the treatment and outcome. Despite ${\bf Z} = \{W_1, W_2, W_3\}$ fulfilling all the critical assumptions set up in Theorem \ref{thm:identification}, our method will incorrectly conclude that there is no valid adjustment set because there exists no $W \in {\bf W}$ that can be used as an auxiliary variable to test C3 and C4. This same limitation exists in the method proposed by \cite{entner2013data} in settings without missing data.

\begin{figure}[ht]
    \centering
    \scalebox{0.8}{
        \begin{tikzpicture}[>=stealth, node distance=1.5cm]
            \tikzstyle{square} = [draw, thick, minimum size=1.0mm, inner sep=3pt]
            \begin{scope}
                \path[->, very thick]
                node[] (a) {$A$}
                node[right of=a] (y) {$Y^{(1)}$}
                node[right of=y] (ry) {$R_Y$}
                node[right of=ry] (i) {$I$}
                node[above of=a, xshift=0.75cm, yshift=0.5cm] (w1) {$W_1$}
                node[right of=w1, xshift=-0.375cm] (u1) {$U_1$}
                node[above of=y, xshift=1.5cm] (w2) {$W_2$}
                node[below of=w1, yshift=0.5cm] (u2) {$U_2$}
                node[above of=a, xshift=-0.6cm, yshift=-0.5cm] (w3) {$W_3$}
                
                (a) edge[blue] (y)
                (y) edge[blue] (ry)
                (i) edge[blue] (ry)
                (u1) edge[blue] (w1)
                (u1) edge[blue] (w2)
                (u2) edge[blue] (w1)
                (u2) edge[blue] (a)
                (w1) edge[blue] (ry)
                (w2) edge[blue] (y)
                (w3) edge[blue] (a)
                (w3) edge[blue] (y)
                ;
            \end{scope}
        \end{tikzpicture}
    }
    \caption{Setting where identification is possible using ${\bf Z} = \{W_1, W_2, W_3\}$, but there are no observed independence constraints that our method can use to verify this.}
    \label{fig:method_fail}
\end{figure}
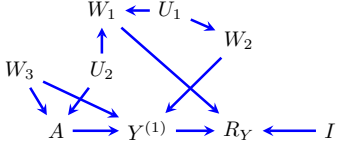

Furthermore, the incentive variable $I$ is assumed to be a randomized variable in our method. In existing observational studies, it may not always be possible to identify such a variable. However, it may be possible to relax this condition such that $I$ is conditionally randomized based on a set of pre-treatment covariates. This could potentially allow for more flexibility in the identification of a valid incentive variable, but we do not pursue this line of inquiry here.

\section{Estimation Procedure}
\label{sec:estimation}

The first step in a practical estimation procedure is to verify C1. If this test fails, then we are unable to empirically verify the validity of candidate adjustment sets using the given incentive $I$. If it succeeds, Algorithm~\ref{alg:search} then searches over all possible assignments for $W$ and ${\bf Z}$ to see if there is a combination of assignments that fulfills C2-C4. The search in Algorithm~\ref{alg:search} proceeds in a style similar to the PC algorithm \citep{spirtes2000causation} for causal discovery (and the order of tests described for Figure~\ref{fig:method_example}) where conditioning sets of size $0$ are tested first, followed by sets of size $1$, and so on. Such an exponential search is unavoidable in situations where no valid adjustment set exists.\footnote{If desired, one may also perform a search for candidate adjustment sets using only a subset of all possible subsets of ${\bf W}$.} We use likelihood ratio tests for all tests in Algorithm~\ref{alg:search}; however, suitable non-parametric tests such as kernel conditional independence tests can also be applied \citep{zhang2011kernel}. If Algorithm~\ref{alg:search} returns a set ${\bf Z}$, then we have found a set satisfying B1, B2, S1, and S2, and we may use $A$ as a valid shadow variable and ${\bf Z}$ as a valid  adjustment set. Otherwise, we conclude that no adjustment set could be validated using our tests.
\begin{algorithm}[t]
\caption{for finding a valid adjustment set {\bf Z}.}\label{alg:search}
\begin{algorithmic}[1]
\For{each $W \in {\bf W}$}
    \State ${\bf Z}_f \gets {\bf W} \setminus \{W\}$
    \For{$i$ from $0$ to $|{\bf Z}_f|$}
        \State $\mathbb{Z}_s$ $\gets $ all possible subsets of ${\bf Z}_f$ with size $i$
        \For{each ${\bf Z} \in \mathbb{Z}_s$}
            \If{C2-C4 are true using $W$ and {\bf Z}}
                \State return ${\bf Z}$
            \EndIf
        \EndFor
    \EndFor
\EndFor
\State return ``no adjustment set found''
\end{algorithmic}
\end{algorithm}

If we find a valid set {\bf Z}, we then estimate the distribution $p(A \mid {\bf Z})$ using the full dataset as it involves only fully observed variables. As the treatment is typically binary, we use a logistic regression model in our procedure. More flexible models such as generalized additive models or random forests are also possible depending on the sample size.

Next, we estimate the propensity score for $R_Y$, $p(R_Y=1 \mid Y^{(1)}, {\bf Z})$. Let the cardinality of ${\bf Z}$ be $|{\bf Z}|=k$. In our estimation procedure, we use the odds ratio factorization of the propensity score in \eqref{eq:or_factorization} with the parameterizations $\pi_0({\bf Z})=\text{expit}(\beta_1 Z_1+\beta_2 Z_2 \dots \beta_k Z_k)$ and $\eta(Y^{(1)},{\bf Z})=\text{exp}(\gamma Y^{(1)})$\footnote{Although we use specific parameterizations for the propensity score and odds ratio in our estimation procedure, our identification strategy is non-parametric. Different strategies for estimating the odds ratio is given by \cite{tchetgen2010doubly}.}. In total, we estimate $k+1$ parameters to recover the propensity score of $R_Y$; therefore, we need $k+1$ mean-zero estimating equations of the form $\E[(\frac{R_Y}{p(R_Y=1 \mid Y^{(1)}, {\bf Z}; {\boldsymbol \beta}, \gamma)}-1) \times h(A, {\bf Z})]$. A simple choice for the first $k$ equations is to use $h(A, {\bf Z})=Z_i$ for each $Z_i \in {\bf Z}$. For the final equation, we use $h(A, {\bf Z}) = \overline{A}$, where $\overline{A}$ is the mean of the variable $A$. This gives us the following system of equations:
\begin{align}
    \E \begin{bmatrix}\bigg(\frac{R_Y}{p(R_Y=1 \mid Y^{(1)}, {\bf Z}; {\boldsymbol \beta}, \gamma)} - 1\bigg) \begin{bmatrix} Z_1 \\ Z_2 \\ \dots \\ Z_k \\ \overline{A} \end{bmatrix} \end{bmatrix} &= 0
    \label{eq:shadow_equations}
\end{align}
Because the expression $\frac{R_Y}{p(R_Y=1 \mid Y^{(1)}, {\bf Z}; {\boldsymbol \beta}, \gamma)}$ is $0$ for each individual in the dataset where $R_Y=0$, the system of estimating equations only uses observed rows of $Y^{(1)}$.

Using the estimated parameters for $p(A \mid {\bf Z})$ and $p(R_Y=1 \mid Y^{(1)}, {\bf Z})$, we get estimates of both propensity scores for each row of data. Since inverse probability weighting estimators can be unstable due to large weights, we clip these propensity scores between values $p_{low}=0.01$ and $p_{high}=0.99$ \citep{hernan2010causal, crump2009dealing}. Finally, we estimate the causal effect by taking the empirical average for the expectation shown in identifying functional \eqref{eq:identification} for treatment values $a$ and $a'$.

A summary of our procedure is as follows: (i) Test C1, if it holds, proceed to (ii), else, terminate the algorithm. (ii) Test C2-C4 using Algorithm~\ref{alg:search}, if it returns a set ${\bf Z}$, proceed to (iii), else terminate the algorithm. (iii) Estimate the propensity scores for $R_Y$ and $A$ and plug them into an inverse probability weighted estimator for the ACE based on the identifying functional in \eqref{eq:identification}.

\subsubsection*{Complexity of the Search Procedure}

The worst-case computational complexity of the proposed method is exponential due to the number of subsets that are considered when testing conditions C1 through C4. However, in practice, the complexity may not necessarily pose any significant challenges to applying this method. According to a meta study that evaluated studies in applied health research that used DAGs, most of these studies only use DAGs with roughly twelve variables \citep{tennant2021use}. DAGs of this size and those of similar sizes should pose no computational issues for our method. For high-dimensional settings, researchers typically rely on some sparsity assumptions such as limiting the maximum size of the conditioning set (as in causal discovery applications) or, alternatively, finding a low-dimensional representation of the high-dimensional confounders that is sufficient for adjustment \citep{ma2019robust}. It is possible to apply such methods in conjunction with ours to deal with self-censoring in the high-dimensional case.

\section{Simulation Study}
\label{sec:experiments}

For our simulations, we generate data according to the graph shown in Figure \ref{fig:experiment_graph} and modifications of it that violate the shadow variable or backdoor conditions. We generate the pre-treatment covariates ${\bf W}$ from a multivariate normal distribution with mean ${\bf 0}$ and covariance matrix $\mathbf{\Sigma}$ such that the off-diagonal entries, corresponding to the covariance of error terms, are non-zero for the pairs $(W_2, W_3)$, $(W_2, W_4)$, and $(W_3, W_4)$.
This is equivalent to a structural equation model with correlated errors where we have unmeasured confounding of the form $W_2 \biedge W_3$, $W_2 \biedge W_4$, and $W_3 \biedge W_4$.
We generate $A, Y^{(1)}$, and $R_Y$ as binary variables as functions of their parents in Figure \ref{fig:experiment_graph}. To generate $R_Y$, we use the odds ratio factorization and parameterization specified in Section~\ref{sec:estimation}. The incentive variable $I$ is normally distributed with a mean of $0$ and variance $2$. 
Precise details of the data generating process are provided in the Appendix.
\begin{figure}[ht]
    \centering
    \scalebox{0.8}{
	    \begin{tikzpicture}[>=stealth, node distance=1.5cm]
			\tikzstyle{square} = [draw, thick, minimum size=1.0mm, inner sep=3pt]
			\begin{scope}
				\path[->, very thick]
				node[] (a) {$A$}
				node[right of=a] (y) {$Y^{(1)}$}
				node[right of=y] (ry) {$R_Y$}
                node[above of=a, xshift=-0.75cm] (w1) {$W_1$}
                node[right of=w1] (w2) {$W_2$}
                node[right of=w2] (w3) {$W_3$}
                node[right of=w3] (w4) {$W_4$}
                node[right of=ry] (i) {$I$}
				
				(a) edge[blue] (y)
                (y) edge[blue] (ry)
                (w1) edge[blue] (a)
                (w2) edge[blue] (a)
                (w2) edge[blue] (y)
                (w2) edge[blue] (ry)
                (w3) edge[blue] (a)
                (w3) edge[blue] (y)
                (w3) edge[blue] (ry)
                (w4) edge[blue] (a)
                (w4) edge[blue] (y)
                (w4) edge[blue] (ry)
                (w2) edge[red, <->, bend left] (w3)
                (w3) edge[red, <->, bend left] (w4)
                (w2) edge[red, <->, bend left, in=125, out=55] (w4)
                (i) edge[blue] (ry)
				;
			\end{scope}
		\end{tikzpicture}
    }
    \caption{Graph used in our simulations.}
    \label{fig:experiment_graph}
\end{figure}
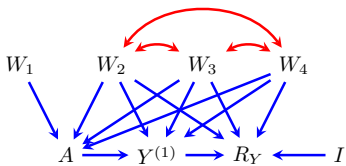

Our first set of experiments focuses on evaluating the effectiveness of Algorithm~\ref{alg:search} for finding a valid adjustment set when such a set exists and to correctly identify that no valid adjustment set exists when no such set exists. Before describing the experiments, we define key terms for evaluating accuracy. A {\it true positive} (TP) occurs when an adjustment set exists and the algorithm identifies the correct set of covariates. If the algorithm does not find an adjustment set or returns an incorrect one, then this is a {\it false negative} (FN). A {\it true negative} (TN) occurs when no possible adjustment set exists and the algorithm correctly finds no adjustment set. If the algorithm detects an adjustment set when no such set exists, this is considered  a {\it false positive} (FP). {\it Sensitivity} is defined as $\frac{\# \text{TP}}{\# \text{TP} + \# \text{FN}}$, and {\it specificity} is defined as $\frac{\# \text{TN}}{\# \text{TN} + \# \text{FP}}$.

The experiment proceeds as follows. We first run 200 trials with data generated according to Figure~\ref{fig:experiment_graph} where Algorithm~\ref{alg:search} should return ${\bf Z} = \{W_2, W_3, W_4\}$. We then run 200 trials where the algorithm should return no adjustment set as no valid set exists using data generated according to a DAG where we add edge $A \diedgeright R_Y$ to Figure~\ref{fig:experiment_graph} with probability $0.5$ or where we treat $W_4$ as a latent variable with probability $0.5$. The size of the dataset for each of these trials ranges from 500 to 10,000, and we use a significance level of $\alpha=0.05$ for our tests. However, note that the effective sample size for some tests is roughly $60\%$ of the full sample size due to the missingness of the outcome. Table \ref{tab:covariate_search} summarizes the results. In the Appendix, we also report results of this experiment using $\alpha=0.01$ and $\alpha=0.1$.
\begin{table}[ht]
    \centering
    \begin{tabular}{|c|c|c|}
    \hline
    \begin{tabular}[c]{@{}c@{}}{\bf Sample} {\bf Size}\end{tabular} & \multicolumn{1}{l|}{{\bf Sensitivity}} & {\bf Specificity} \\ \hline
    500   & 0.011 & 0.350 \\ \hline
    2500  & 0.577 & 0.556 \\ \hline
    5000  & 0.812 & 0.818 \\ \hline
    10000 & 0.930 & 0.925 \\ \hline
    \end{tabular}
    \caption{Results of covariate search experiment.}
    \label{tab:covariate_search}
\end{table}

The second set of experiments uses data from Figure~\ref{fig:experiment_graph} to evaluate the effectiveness of our method for downstream estimation. We compare these estimates to the bias introduced by either failing to adjust for missing data but using the correct backdoor adjustment set, or by failing to use a valid backdoor adjustment set but correctly adjusting for missing data. To estimate the ACE without adjusting for missing data, we subset the dataset to only rows of data where the outcome is observed and use a standard inverse probability weighting estimator. To estimate the ACE while using an invalid backdoor adjustment set, we use a subset of the correct adjustment set, $\{W_2, W_3\}$, after correctly adjusting for missing data. We compare these two estimates with one obtained by running our full procedure described in Section~\ref{sec:estimation} with $\alpha=0.05$. We also generate estimates obtained using an independence oracle in Algorithm~\ref{alg:search} to emphasize the importance of reliable conditional independence tests. We run 200 trials each for the sample sizes 500, 2,500, 5,000, and 10,000.

\begin{figure*}[ht]
    \centering
    \begin{subfigure}[b]{0.33\textwidth}
        \centering
        \includegraphics[scale=0.33]{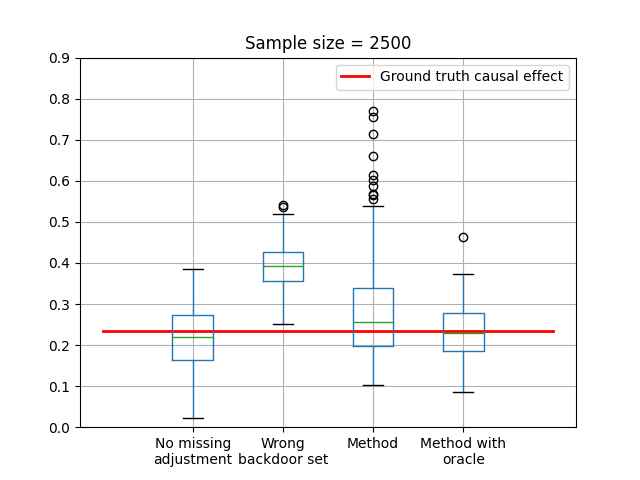}
        \caption{}
    \end{subfigure}
    \begin{subfigure}[b]{0.33\textwidth}
        \centering
        \includegraphics[scale=0.33]{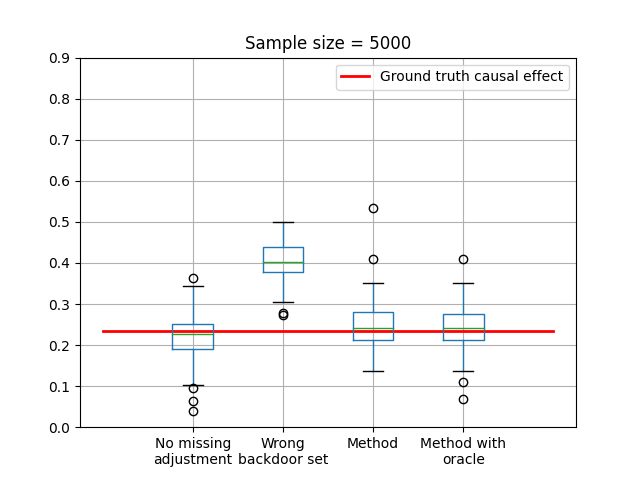}
        \caption{}
    \end{subfigure}
    \begin{subfigure}[b]{0.33\textwidth}
        \centering
        \includegraphics[scale=0.33]{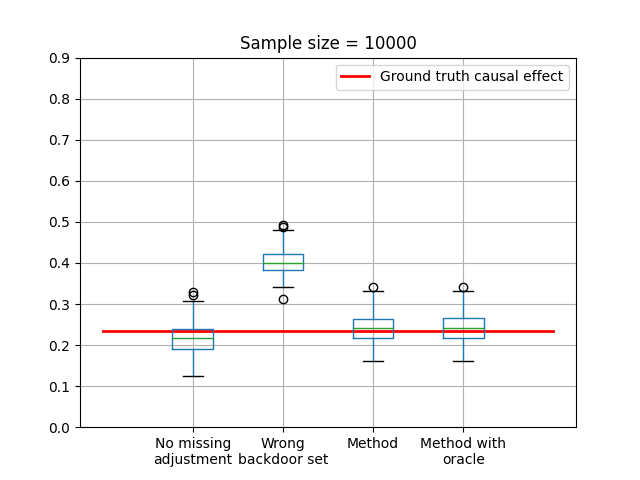}
        \caption{}
    \end{subfigure}
    \caption{Simulation results for estimating the average causal effect using different methods and sample sizes.}
    \label{fig:experiments}
\end{figure*}

Results for 500 samples exhibit high bias due to inaccuracy of the tests and are reported in the Appendix in the interest of space. The red line in Figure \ref{fig:experiments} shows the ground truth causal effect. At sample size 10,000, the search algorithm correctly identifies the adjustment set for $93\%$ of the trials, so the estimates from the method are nearly identical to the estimates from using the independence oracle. As the sample size increases, estimation from failing to adjust for missing data and confounding bias converge to biased values for the ACE. The method that ignores missingness produces reasonable estimates at low sample sizes, but it shows asymptotic convergence to a biased estimate. As expected, estimates obtained by using our full pipeline converge to the ground truth causal effect as sample size increases. Python code implementing our estimation procedure and to reproduce our numerical studies can be found online: \url{https://github.com/jacobmchen/mnar-recoverability}.

\section{Conclusion}
\label{sec:conclusion}

In this paper, we discuss methods for causal effect estimation with a self-censoring outcome and when the underlying causal structure of the graph is unknown. We prove that when a set ${\bf Z} \subset {\bf W}$ satisfies both the backdoor adjustment set and the shadow variable adjustment set criteria, then identification of the average causal effect is possible through an inverse probability weighting functional. We further describe a series of tests that may be used to empirically identify such a valid set ${\bf Z}$ using the observed data. We present a simple search algorithm that uses our tests as a subroutine and estimates the average causal effect using an inverse probability weighting estimator if a valid set ${\bf Z}$ is found. Finally, we conclude with experiments based on synthetic data that demonstrate the accuracy of our search algorithm and estimation procedure.

We have thus made progress on identifying situations where it is possible to overcome self-censoring and confounding to compute an unbiased estimate for the average causal effect. To the best of our knowledge thus far in the literature, all methods for covariate selection under MNAR data do not allow for self-censoring on the outcome and require prior knowledge of the underlying causal structure.

As self-censoring is a difficult problem, however, many open questions remain: To what extent can we relax the assumption that the incentive $I$ must be randomized? Can we design semiparametric estimation strategies for the tests and final estimation piece that exhibit desirable statistical properties, such as robustness to model mispecification and lower asymptotic variance? Is it possible to apply different identification strategies, such as frontdoor adjustment proposed by \cite{pearl1995causal}, to identify the average causal effect under self-censoring of the outcome when backdoor adjustment is not applicable? These issues and open questions may be the focus of future research.


\begin{acknowledgements}
DM was partially supported by the National Institutes of Health under award number K25ES034064 from NIEHS.  
\end{acknowledgements}

\clearpage

\bibliography{references}

\clearpage

\appendix

\onecolumn

{\Large \bf APPENDIX}

\section{Specifics of Data Generating Process}

For our simulations, we generate data according to the graph shown in Figure \ref{fig:experiment_graph} and modifications of it that violate the shadow variable or backdoor conditions. We generate the pre-treatment covariates ${\bf W}$ from a multivariate normal distribution with mean {\bf 0} and covariance matrix $\mathbf{\Sigma} =$
\begin{align*}
\begin{bmatrix}
    1.2 & 0 & 0 & 0\\
    0 & 1 & 0.4 & 0.4\\
    0 & 0.4 & 1 & 0.3\\
    0 & 0.4 & 0.3 & 1
    \end{bmatrix}.
\end{align*}
The above data generating process is equivalent to a structural equation model with correlated errors due to unmeasured confounders between the pairs $(W_2, W_3)$, $(W_2, W_4)$, and $(W_3, W_4)$.   We generate $A, Y^{(1)}, I$, and $R_Y$ according to structural equation models following edges in Figure \ref{fig:experiment_graph}. Note that we also clip all probabilities to be between the ranges of $0.01$ and $0.99$. 

We generate $A$ as a binary variable with the following probabilities:
\begin{align*}
    p(A=1 \mid W_1, W_2, W_3, W_4) &= \text{expit}(0.52 + 2*W_1 + 2*W_2 + 2*W_3 + 2*W_4) \\
    p(A=0 \mid W_1, W_2, W_3, W_4) &= 1-p(A=1 \mid W_1, W_2, W_3, W_4)
\end{align*}
Next, $Y^{(1)}$ is generated similarly with the following probabilities:
\begin{align*}
    p(Y^{(1)}=1 \mid A, W_2, W_3, W_4) &= \text{expit}(3*A + 2*W_2 + 2*W_3 + 2*W_4) \\
    p(Y^{(1)}=0 \mid A, W_2, W_3, W_4) &= 1-p(Y^{(1)}=1 \mid A, W_2, W_3, W_4)
\end{align*}
The variable $I$ is simply a random normal variable with mean $0$ and variance $2$, i.e. $I \sim \mathcal{N}(0, 2)$.

We use the odds ratio parameterization to generate $R_Y$ with the following two probabilities. We first specify $p(R_Y=1 \mid Y^{(1)}=0, {\bf W} \setminus \{W_1\}, I)$, which represents the probability of $R_Y=1$ when $Y^{(1)}$ is at its chosen reference value of $0$. We then use that probability to generate $p(R_Y=1 \mid Y^{(1)}, {\bf W} \setminus \{W_1\}, I)$ at all values of $Y^{(1)}$.
\begin{align*}
    p(R_Y=1 \mid Y^{(1)}=0, {\bf W} \setminus \{W_1\}, I) &= \text{expit}(W_2 + W_3 + W_4 + 0.5*I)
\end{align*}
\begin{align*}
    p(R_Y&=1 \mid Y^{(1)}, {\bf W} \setminus \{W_1\}, I) = \\
    &\frac{p(R_Y=1 \mid Y^{(1)}=0, {\bf W} \setminus \{W_1\}, I)}{p(R_Y=1 \mid Y^{(1)}=0, {\bf W} \setminus \{W_1\}, I) + \text{exp}(-1.5*Y^{(1)}) \times (1-p(R_Y=1 \mid Y^{(1)}=0, {\bf W} \setminus \{W_1\}, I))}.
\end{align*}

In the case where we add $A \diedgeright R_Y$ to Figure~\ref{fig:experiment_graph}, we add $1.5*A$ in the expit function for $p(R_Y=1 \mid Y^{(1)}=0, \cdot)$.

\clearpage

\section{Proof of Theorem \ref{thm:identification}}

We first note that the presence of $R_Y$ in the numerator ensures that we only use observed rows of data. Further, the propensity scores in the denominator are identified: $p(A\mid {\bf Z})$ only depends on observed quantities, and $p(R_Y=1\mid A, Y^{(1)}, {\bf Z}) = p(R_Y=1\mid Y^{(1)}, {\bf Z})$ is identified using S1, S2, and the completeness condition. We now prove that the proposed identifying functional is equal to the backdoor adjustment functional and counterfactual mean under the full data law.

\begin{proof}
\begin{align*}
    & \E\bigg[\frac{R_Y \times \mathbb{I}(A=a) \times Y}{p(R_Y=1 \mid Y^{(1)}, {\bf Z}) \times p(A=a \mid {\bf Z})} \bigg] \\
    &=^{(1)} \sum_{R_Y, Y^{(1)}, A, {\bf Z}, Y} p(R_Y, Y^{(1)}, A, {\bf Z}, Y) \times \frac{R_Y \times \mathbb{I}(A=a) \times Y}{p(R_Y=1 \mid Y^{(1)}, {\bf Z}) \times p(A=a \mid {\bf Z})} \\
    &=^{(2)} \sum_{Y^{(1)}, A, {\bf Z}} p(R_Y=1, Y^{(1)}, A, {\bf Z}) \times \frac{\mathbb{I}(A=a) \times Y^{(1)}}{p(R_Y=1 \mid Y^{(1)}, {\bf Z}) \times p(A \mid {\bf Z})} \\
    &=^{(3)} \sum_{Y^{(1)}, A, {\bf Z}} p(R_Y=1 \mid Y^{(1)}, A, {\bf Z}) p(Y^{(1)} \mid A, {\bf Z}) p(A \mid {\bf Z}) p({\bf Z}) \times \frac{\mathbb{I}(A=a) \times Y^{(1)}}{p(R_Y=1 \mid Y^{(1)}, {\bf Z}) \times p(A=a \mid {\bf Z})} \\
    &=^{(4)} \sum_{Y^{(1)}, {\bf Z}} p(R_Y=1 \mid Y^{(1)}, {\bf Z}) p(Y^{(1)} \mid A=a, {\bf Z}) p(A=a \mid {\bf Z}) p({\bf Z}) \times \frac{Y^{(1)}}{p(R_Y=1 \mid Y^{(1)}, {\bf Z}) \times p(A=a \mid {\bf Z})} \\
    &=^{(5)} \sum_{Y^{(1)}, {\bf Z}} p(Y^{(1)} \mid A=a, {\bf Z}) \times p({\bf Z}) \times Y^{(1)} \\
    &=^{(6)} \sum_{{\bf Z}} \E[Y^{(1)} \mid A=a, {\bf Z}] \times p({\bf Z}) =^{(7)} \E[Y^{(a, 1)}].
\end{align*}
In (1) we apply the law of the unconscious statistician; in (2) we evaluate the sum over $R_Y$ and use missing data consistency; in (3) we apply the chain rule of probability; in (4) we evaluate the sum over $A$ and drop $A$ from the propensity score of $R_Y$ due to condition S2; (5) follows from cancellation of common terms in the numerator and denominator; (6) follows from definition of expectation; the last step (7) follows from the fact that ${\bf Z}$ satisfies the backdoor conditions B1 and B2.
\end{proof}

\section{Proof of Equation \ref{eq:or_factorization}}

For completeness, we provide a proof for odds ratio parameterization of the propensity score in \eqref{eq:or_factorization}. First, from \cite{chen2007semiparametric} we have an odds ratio factorization of the joint distribution $p(R_Y, Y^{(1)} \mid {\bf Z})$ as follows,
\begin{align}
    p(R_Y, Y^{(1)} \mid {\bf Z}) &= \frac{p(R_Y \mid Y^{(1)}=y_0, {\bf Z}) \times p(Y^{(1)} \mid R_Y=1, {\bf Z}) \times \odds(Y^{(1)}, R_Y \mid {\bf Z})}{\sum_{R_Y, Y^{(1)}} p(R_Y \mid Y^{(1)}=y_0, {\bf Z}) \times p(Y^{(1)} \mid R_Y=1, {\bf Z}) \times \odds(Y^{(1)}, R_Y \mid {\bf Z})},
    \label{eq:chen_or_factorization}
\end{align}
where $y_0$ is a reference value for $Y^{(1)}$ and $1$ is the reference value for $R_Y$, and the denominator of is a normalizing function. Let $\psi = p(R_Y \mid Y^{(1)}=y_0, {\bf Z}) \times p(Y^{(1)} \mid R_Y=1, {\bf Z}) \times \odds(Y^{(1)}, R_Y \mid {\bf Z})$ and $\psi_1 = \left. \psi \right|_{R_Y=1}$. As in the main section of the paper, $\pi_0 \coloneqq p(R_Y =1 \mid Y^{(1)} = y_0, {\bf Z})$ and $\eta(Y^{(1)}, {\bf Z}) \coloneqq \odds(R_Y=0, Y^{(1)} \mid {\bf Z})$. We present the proof and an explanation of each step below.

(1) and (2) follow from standard laws of probability.
In (3), we apply the odds ratio factorization in \eqref{eq:chen_or_factorization} to both the numerator and the denominator. In (4), we cancel out like terms in both the numerator and denominator. In (5), we simply expand out $\psi_1$ and $\psi$ according to our previous definitions of these two terms. In (6), we note that $\odds(Y^{(1)}, R_Y=1 \mid {\bf Z})$ has $R_Y$ at its reference value of $1$; hence, it is equal to $1$. Further, we move the term $p(Y^{(1)} \mid R_Y=1, {\bf Z})$ outside of the sum in the denominator because this term is not a function of $R_Y$. In (7), we cancel out like terms from the numerator and denominator. Finally, in (8), we explicitly write out the sum over $R_Y$, which has only two possible values. When $R_Y=1$, $R_Y$ is at its reference value in the odds ratio, so the odds ratio term disappears, and we are just left with $\pi_0({\bf Z})$. When $R_Y=0$, we know that $p(R_Y=0 \mid Y^{(1)}=y_0, {\bf Z}) = 1-\pi_0({\bf Z})$ and that neither $Y^{(1)}$ nor $R_Y$ are at their reference values in the odds ratio term. Therefore, the odds ratio term remains. 

\begin{proof}
    \begin{align*}
        p(R_Y=1 \mid Y^{(1)}, {\bf Z}) &=^{(1)} \frac{p(R_Y=1, Y^{(1)} \mid {\bf Z})}{p(Y^{(1)} \mid {\bf Z})} \\
        &=^{(2)} \frac{p(R_Y=1, Y^{(1)} \mid {\bf Z})}{\sum_{R_Y} p(R_Y, Y^{(1)} \mid {\bf Z})} \\
        &=^{(3)} \frac{\frac{\psi_1}{\sum_{R_Y, Y^{(1)}} \psi}}{\frac{\sum_{R_Y} \psi}{\sum_{R_Y, Y^{(1)}} \psi}} \\
        &=^{(4)} \frac{\psi_1}{\sum_{R_Y} \psi} \\
        &=^{(5)} \frac{p(R_Y=1 \mid Y^{(1)}=y_0, {\bf Z}) \times p(Y^{(1)} \mid R_Y=1, {\bf Z}) \times \odds(Y^{(1)}, R_Y=1 \mid {\bf Z})}{\sum_{R_Y} p(R_Y \mid Y^{(1)}=y_0, {\bf Z}) \times p(Y^{(1)} \mid R_Y=1, {\bf Z}) \times \odds(Y^{(1)}, R_Y \mid {\bf Z})} \\
        &=^{(6)} \frac{p(R_Y=1 \mid Y^{(1)}=y_0, {\bf Z}) \times p(Y^{(1)} \mid R_Y=1, {\bf Z})}{p(Y^{(1)} \mid R_Y=1, {\bf Z}) \times \sum_{R_Y} p(R_Y \mid Y^{(1)}=y_0, {\bf Z}) \times \odds(Y^{(1)}, R_Y \mid {\bf Z})} \\
        &=^{(7)} \frac{p(R_Y=1 \mid Y^{(1)}=y_0, {\bf Z})}{\sum_{R_Y} p(R_Y \mid Y^{(1)}=y_0, {\bf Z}) \times \odds(Y^{(1)}, R_Y \mid {\bf Z})} \\
        &=^{(8)} \frac{\pi_0({\bf Z})}{\pi_0({\bf Z}) + \eta(Y^{(1)}, {\bf Z})(1-\pi_0({\bf Z}))}
    \end{align*}
\end{proof}

\section{Additional Simulation Results}

We report additional simulation results for simulations of the search algorithm described in section \ref{sec:experiments}. Table \ref{tab:covariate_search_pval0.01} gives results when using $\alpha=0.01$ to conclude dependence between two variables.
\begin{table}[ht]
    \centering
    \begin{tabular}{|c|c|c|}
    \hline
    \begin{tabular}[c]{@{}c@{}}{\bf Sample} {\bf Size}\end{tabular} & \multicolumn{1}{l|}{\bf Sensitivity} & {\bf Specificity} \\ \hline
    500   & 0.0 & 0.394 \\ \hline
    2500  & 0.269 & 0.383 \\ \hline
    5000  & 0.690 & 0.717 \\ \hline
    10000 & 0.828 & 0.952 \\ \hline
    \end{tabular}
    \caption{Tables showing the accuracy of tests for different sample sizes and $\alpha=0.01$.}
    \label{tab:covariate_search_pval0.01}
\end{table}

Next, Table \ref{tab:covariate_search_pval0.1} gives the results of the simulations when using $\alpha=0.1$ to conclude dependence between two variables.

\begin{table}[ht]
    \centering
    \begin{tabular}{|c|c|c|}
    \hline
    \begin{tabular}[c]{@{}c@{}}{\bf Sample} {\bf Size}\end{tabular} & \multicolumn{1}{l|}{\bf Sensitivity} & {\bf Specificity} \\ \hline
    500   & 0.022 & 0.363 \\ \hline
    2500  & 0.671 & 0.630 \\ \hline
    5000  & 0.837 & 0.770 \\ \hline
    10000 & 0.949 & 0.857 \\ \hline
    \end{tabular}
    \caption{Tables showing the accuracy of tests for different sample sizes and $\alpha=0.1$.}
    \label{tab:covariate_search_pval0.1}
\end{table}

As the p-value increases, the accuracy of the tests for correctly predicting an adjustment set when one is possible increases. On the other hand, the accuracy of the tests for correctly identifying that there is no possible adjustment set when no such set exists decreases as p-value increases. For all p-values, the accuracy of the tests in general increase as the sample size increases.

\begin{figure}[ht]
    \centering
    \includegraphics[scale=0.6]{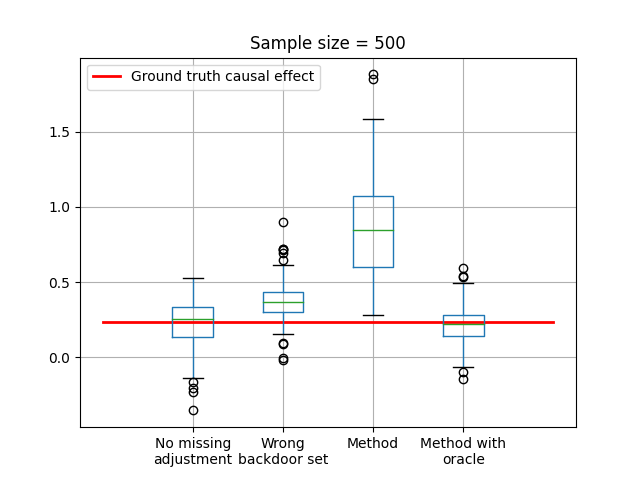}
    \caption{Estimation results for sample size 500.}
    \label{fig:estimation_500}
\end{figure}

Next, we report additional simulation results for estimation of the causal effect described in Section~\ref{sec:experiments}. Figure~\ref{fig:estimation_500} shows the estimation results for sample size 500. When using the correct adjustment method, our practical estimation method is able to accurately recover the causal effect. However, using the full pipeline of our method, the estimates are fairly inaccurate. This is to be expected, as the sensitivity of the covariate search at a small sample size is quite inaccurate regardless of p-value. In addition, due to missing data, the effective sample size of the data is roughly $300$.

\end{document}